\documentclass[letterpaper, 10 pt, conference]{ieeeconf}  %

\usepackage[utf8]{inputenc}

\usepackage{graphicx}
\usepackage{graphics}
\usepackage{amssymb,amsmath,latexsym,amsfonts,amsthm}
\usepackage{booktabs}
\usepackage{array}
\usepackage{soul}
\usepackage{comment}
\usepackage[hyphens]{url}
\graphicspath{{figs/}}
\usepackage[usenames,dvipsnames]{xcolor}
\usepackage{siunitx}
\usepackage{multirow}
\usepackage{subcaption} %
\usepackage{color,soul}
\usepackage{lipsum}
\usepackage{mathtools, cuted}
\usepackage[export]{adjustbox}
\usepackage{dblfloatfix}    %
\usepackage{bbm}
\usepackage{blindtext}
\usepackage{bm}
\usepackage{cite} %
\usepackage[english]{babel}
\usepackage{cleveref}
\usepackage{textgreek} %
\usepackage{mwe}
\usepackage{xcolor}
\usepackage{soul}
\usepackage{savesym} %
\usepackage{siunitx}
\savesymbol{degree}
\usepackage{gensymb} 
\restoresymbol{gensymb}{degree}
\usepackage{tabstackengine}
\stackMath
\setstackgap{L}{16pt} %
\setstacktabbedgap{6pt} %
\def\lrgap{\kern6pt}

\def\xbracketMatrixstack#1{\left[\lrgap\tabbedCenterstack{#1}\lrgap\right]}

\newtheorem{assumption}{Assumption}%
\newtheorem{theorem}{Theorem}%
\newtheorem{lemma}{Lemma}%
\newtheorem{remark}{Remark}%
\newtheorem{proposition}{Proposition}%

\newcommand{\norm}[1]{\left\lVert#1\right\rVert}

\renewcommand{\baselinestretch}{0.97} %
\DeclareCaptionFont{mysize}{\fontsize{10pt}{12pt}\selectfont} %
\usepackage{tikz}
\usepackage{pgfplots}
\pgfplotsset{compat = newest}  %
\usetikzlibrary{positioning,shapes.geometric}
\usetikzlibrary{arrows, arrows.meta}
\usetikzlibrary{patterns,patterns.meta}
\usepackage[outline]{contour}
\contourlength{2.5pt}

\newcommand\epsi{0.05} %
\newcommand\VL{0.025}  %
\newcommand\pr{0.002}  %
\newcommand\prr{0.006}  %
\newcommand\axlim{0.14}
\newcommand\fntsz{\fontsize{15}{0}\selectfont}
\definecolor{green}{gray}{0.8}
\definecolor{blue}{gray}{0.6}
\definecolor{orange}{gray}{0.4}
\definecolor{lightGreen}{HTML}{6fcfff}

\IEEEoverridecommandlockouts  %
\title{\LARGE \bf Parameter Conditions to Prevent Voltage Oscillations Caused by LTC-Inverter Hunting on Power Distribution Grids}

\author{Jaimie Swartz$^{1}$, Federico Celi$^{2}$, Fabio Pasqualetti$^{2}$, and Alexandra von Meier$^{1}$%
\thanks{*This work supported in part by the U.S. Department of Energy, Award DE-EE0008008.}%
\thanks{$^{1}$J. Swartz and A. von Meier are with the Department of Electrical Engineering and Computer Science,
        University of California at Berkeley, Berkeley, California, USA
        {\tt\small \{jaimie.swartz,vonmeier\}@berkely.edu}}%
\thanks{$^{2}$F. Celi and F. Pasqualetti are with the Department of Mechanical Engineering, University of California at Riverside, Riverside, California, USA
        {\tt\small \{fceli,fabiopas\}@engr.ucr.edu}}%
}

\newcommand\arxiv{}

\begin{document}

\maketitle

\thispagestyle{empty} %
\pagestyle{empty} %

\begin{abstract}
    As more distributed energy resources (DERs) are connected to the power grid, it becomes increasingly important to ensure safe and effective coordination between legacy voltage regulation devices and inverter-based DERs. In this work, we show how a distribution circuit model, composed of two LTCs and two inverter devices, can create voltage oscillations even with reasonable choices of control parameters. By modeling the four-device circuit as a switched affine hybrid system, we analyze the system's oscillatory behavior, both during normal operation and after a cyber-physical attack. Through the analysis we determine the specific region of the voltage state space where oscillations are possible and derive conditions on the control parameters to guarantee against the oscillations. Finally, we project the derived parameter conditions onto 2D spaces, and describe the application of our problem formulation to grids with many devices. 
\end{abstract}

\section{Introduction}
The traditional voltage regulation problem is to design load-tap changer (LTC), voltage regulator, and capacitor bank control parameters such that, over a minutes-to-day duration, the distribution grid voltage is kept within 5\% of the nominal voltage (ANSI C84.1 standard) to avoid interrupting or damaging customer equipment. The addition of naively controlled DERs such as solar PV makes this voltage regulation problem harder \cite{Ahmed} and can cause LTCs to actuate much more frequently, reducing their lifespan. Depending on the control logic and design parameters, connecting smart inverters can either alleviate \cite{Chamana} these voltage fluctuations or be a source of adverse interactions \cite{unintended_journal}. 

One type of adverse interaction is \emph{device hunting} which we define as one or more devices actuating in a repeated sequence that results in periodic voltage oscillations. Hunting among LTCs has been observed by utilities since the 1980s \cite{smith_hunting,Hiskens2007} and has been modeled as a hybrid system in the literature \cite{Hiskens2007}. Inverter-based DERs may be able to solve these problems if their control parameters are set appropriately. However, current inverter standards \cite{IEEE_1547} require inverter control parameters to be adjustable by a remote entity's communication network, which introduces a potential vulnerability to cyberattacks \cite{sahoo2019cyber}. Bad parameters sent to inverters on a circuit, whether deliberately or by mistake, can trigger adverse interactions \cite{Ciaran}. This motivates our investigation into how a poor choice of device parameters can lead to hunting, as simulated in Fig. \ref{motex}.

\begin{figure}[!h]
    \centering
    \begin{subfigure}{0.235\textwidth}
         \centering
       \includegraphics[width=\textwidth]{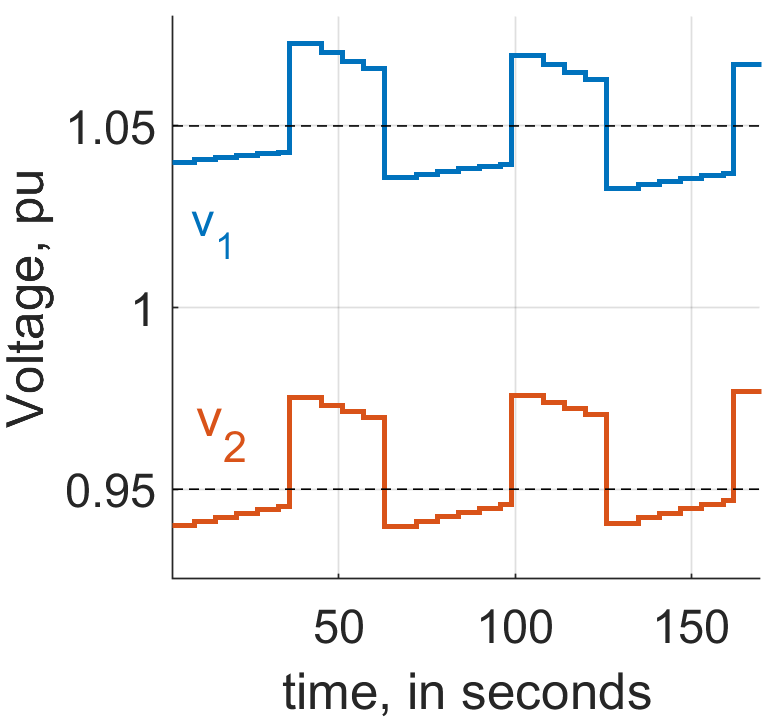}
        \caption{Device hunting for $g=0.5$}
     \end{subfigure}
     \hfill
    \begin{subfigure}{0.235\textwidth}
             \centering
       \includegraphics[width=\textwidth]{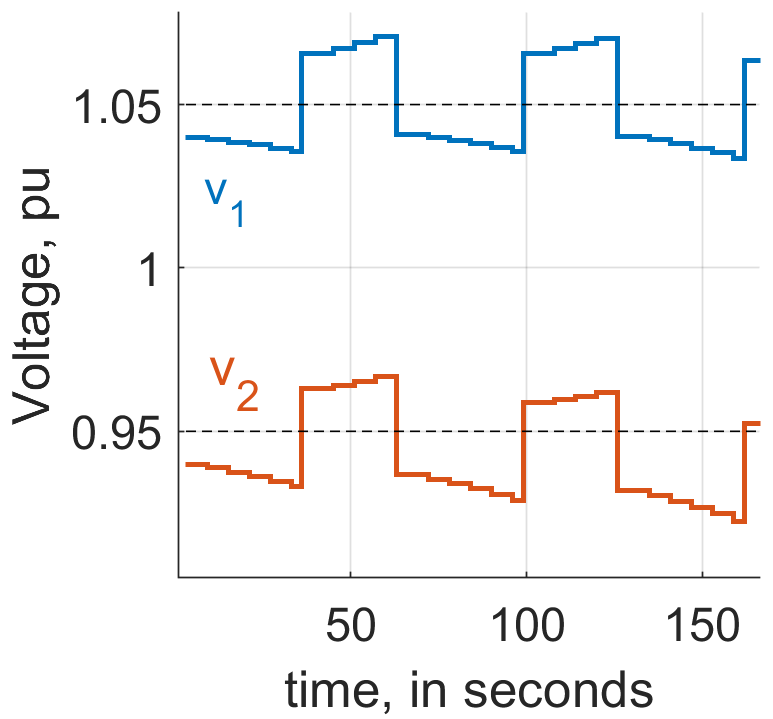}
        \caption{Device hunting for $g=-0.5$}
     \end{subfigure}
     \caption{Simulation of two LTCs and two inverters that create voltage oscillations with a period of 140 seconds. Circuit is in Fig. \ref{LTC_inv_circuit} and parameters are the defaults listed in Table \ref{tab:base_parms}.}
     \label{motex} %
\end{figure}
\emph{Related work:}
There are two common approaches when selecting distribution-grid device parameters for voltage control: an approach based on rules of thumb followed by simulation, and an optimization-based approach. For basic operation of LTCs and parameter rules of thumb we refer to the text \cite[Chapter 9.3]{Gonen_distr_text}, and to papers \cite{Licheng_journal,Ahmed}. One limitation of this approach is the assumed timescale separation between the LTC and inverters, which may not hold when LTCs delays are shortened to handle increased voltage variability \cite{Shahrzad}. Moreover, the simulations in these papers are not sufficient to guarantee against the possibility of sustained voltage oscillations.

Meanwhile, optimization methods focus on addressing the non-convexity of the mixed-integer optimal power flow problem resulting from including both continuous (smart inverter) and discrete (LTC and voltage regulator) dynamics \cite{Christa_opttap,Garcia_opttap}. However, the optimal parameter solution determines tap positions and power dispatch rather than the device parameters in the control law. As such, the optimal solution does not provide significant insights into the symbolic relationship between control parameters and adverse interactions.

Our goal of deriving symbolic parameter conditions for hybrid systems stability is a challenging task and remains relatively unexplored in control systems literature. The design of switching strategies for stabilization commonly assumes one can switch modes at any time, rather than switching according to parameterized conditions \cite{Mignone,Antsaklis}. Linear parameter variation (LPV) literature commonly assumes the parameters to be time-varying, but here our parameters are time-independent and we want to solve for them symbolically \cite{Lim}. Therefore, our methodology leverages our familiarity with the specific system's model dynamics.

\emph{Paper contribution:}
We seek to analyze how parameters of LTCs and inverters prevent or contribute to voltage oscillations created by device hunting. The results yield two distinct benefits: (i) the parameter conditions can provide intuition for improving the rules of thumb used to operate these devices in industry and (ii) the conditions can be directly checked on each device without simulations, enabling real-time checks on incoming parameter updates by remote entities.

\emph{Paper Organization:} The remainder of the paper is organized as follows. In Sec. \ref{probform1_section}, the system equations are presented and the parameter condition problem is solved for two subsystems. In Sec. \ref{probform2_section}, the parameter condition problem is solved for the full system. In Sec. \ref{results_section} the conditions on the full system conditions are illustrated. In Sec. \ref{conclusion_section} we present conclusions and future work. %
\ifdefined\conference %
    \input{input_files/ref_to_arxiv}
    
\section{System Dynamics and Subsystem Analysis}
\label{probform1_section}

\subsection{Notation}
Let $\mathbb{R}^c$ be the $c$-dimensional vectors of real numbers. Given a time-varying vector $x$, we let $x[k]$ be its value at time $k$, and $x_i[k]$ be the value of the $i$-th element of $x$ at time $k$. Let $x^\top$ be the transpose of vector $x$. Let $||.||_i$ denote the $\ell_i$ vector norm for $i=1,2,\infty$. We define the \emph{margin} $M(S,d)$ as the set of points $x$ that are within distance $d$ from a set $S$. That is, $M(S,d) \coloneqq \left \{ x \notin S : \exists ~ s_i \in S : \norm{x - s_i}_2 \leq d \right \}$. We define a \emph{partition} of a set $S$ to be a collection of subsets $S_i$, $i=1,...k$ such that $\bigcup_{i=1,...,k} S_i=S$ and $S_i \cap S_j=\emptyset$ for all $i\neq j$. We use $\rightarrow$ and $\nrightarrow$ to indicate possible and impossible transitions between partitions, respectively. Let $S^\prime$ denote the complement of set $S$, and $S \setminus Q$ be the part of the set $S$ that is not in $Q$.

\subsection{Overview of Four-Device System}

\begin{figure}[!h]  
   \centering 
   \includegraphics[width=.47\textwidth]{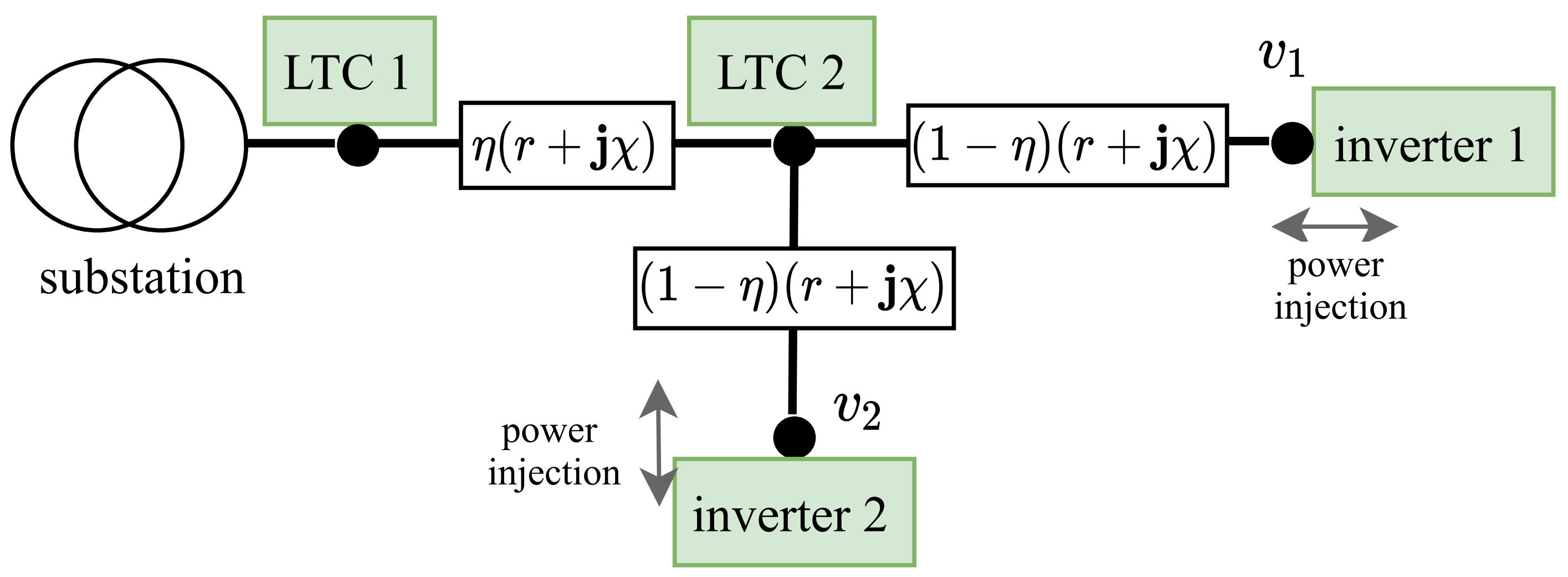}
    \caption{Distribution circuit with LTC and inverter devices}
    \label{LTC_inv_circuit} 
\end{figure}

We model four devices operating on the radial circuit shown in Fig. \ref{LTC_inv_circuit}. 
 Suppose there are constant loads at all nodes shown. LTC1 and inverter1 regulate voltage $v_1$, while LTC2 (sometimes called a line voltage regulator) and inverter2 regulate voltage $v_2$. We let $v^{\text{diff}}\coloneqq v_1-v_2$. The LTCs estimate these voltages using line drop compensation, and the inverters measure these voltages with their internal sensors. All devices operate with fixed time delay logic; that is, they only respond when the voltage remains outside the deadband for a certain delay $d$. 
We denote the upper deadband boundary as $v^+=v^{ref}+\varepsilon$, and lower deadband boundary as $v^-=v^{ref}-\varepsilon$. Then we define the deadband
\begin{equation}
    D\coloneqq \{(v_1,v_2): v^-\leq v_1 \leq v^+,~v^-\leq v_2\leq v^+\}.\\ %
\end{equation}
Because there are no active dynamics when both voltages are in the deadband, the set $D$ is invariant. When $v_i> v^+$ we say $v_i$ has an \emph{overvoltage}, and when $v_i<v^-$ we say $v_i$ has an \emph{undervoltage}. With the shorthand $v_i \in D$ or $v_i \notin D$, we refer to whether $v^-<v_i<v^+$ is satisfied or not.

Next, we define normal operating $W\coloneqq (H \cap P)$ where %
\begin{align} 
  \phantom{P} %
  &\begin{aligned}
    \mathllap{H} &\coloneqq \{(v_1,v_2): (v_1>v^- ~\text{and}~ v_2>v^-) ~\text{or}~\\& (v_1<v^+ ~\text{and}~ v_2<v^+)\}
  \end{aligned} \label{H} \\
  &\begin{aligned}
    \mathllap{P} &\coloneqq \{(v_1,v_2): \norm{v[k] - v^{ref}}_\infty<3\varepsilon\}. %
  \end{aligned}\label{P}
\end{align}
The hourglass-shaped set $H$ disallows one voltage from being above $v^+$ when the other is below $v^-$. $P$ bounds the distance each voltage can be away from $v^{ref}$.

\begin{assumption}(normal operation)
    A given system $\Sigma$ operates in normal operating states $W$ defined by \eqref{H} and \eqref{P}. \label{assump:ov_uv}
\end{assumption}
\vspace{-0.1in}
We consider states outside $W$ to be abnormal operating conditions that should be addressed with the grid's protection system rather than the system dynamics analyzed in this work. We are interested in the system behavior while $v \in W$, especially if the initial condition (IC) starts in $W$ and eventually leaves $W$. 

 The goal is to coordinate the device actions so that $v_1$ and $v_2$ land inside the deadband without device hunting. Table \ref{tab:base_parms} summarizes the notation for the states, fixed parameters, and symbolic parameters we are interested in designing to guarantee against hunting. The table also assumes relationships from basic operation of power systems that are drawn from \cite[Chapter 9.3]{Gonen_distr_text} and \cite{Guannan}. In the table, all fixed and symbolic variables except for $g$ are taken to be positive and real valued due to their physical meaning. 

\begin{table}[htbp]
\caption{\label{tab:base_parms} Notation and variable types }
\begin{center}
\begin{tabular}{|p{0.7cm}|p{1cm}|p{2.1cm}|p{0.7cm}|p{1.9cm}|}    \hline
\textbf{} & \textbf{$\quad$ type} & \textbf{description} & \textbf{default} & \textbf{relationship}\\ \hline
  $v^{ref}$ & fixed & voltage ref. (p.u.) & 1.0   & --\\ \hline
   $\varepsilon$ & fixed & half of deadband width & 0.1  &  -- \\ \hline
 $v^-$, $v^+$ & fixed & deadband boundary & 0.95, 1.05 & $v^-=v^{ref}-\varepsilon$, $v^+=v^{ref}+\varepsilon$ \\ \hline
 $\chi$ & fixed & line reactance to the substation (p.u.) & 0.1 & --\\ \hline
$\eta$ & fixed &  impedance damping factor & 0.9 &  $0<\eta<1$\\ \hline
 $v_{10}$ & state & node 2 initial voltage (p.u.) & 1.04 & -- \\ \hline
 $v_{20}$ & state & node 3 initial voltage (p.u.) & 0.94 & --\\ \hline
 $d_{inv}$ &  symbolic & inverter 1 and 2 delay (s) & 4 & $d_{inv}<d_{L1}$\\ \hline
 $d_{L1}$ & symbolic & LTC1 delay (s) & 30 & --\\ \hline
 $d_{L2}$ & symbolic & LTC2 delay (s) & 40 & $d_{L1}<d_{L2}$, $d_{L2}<2d_{L1}$\\ \hline
 $\bar{v}_L$ & symbolic & tap voltage (p.u.)& 0.03 & $\bar{v}_L<2\varepsilon$\\ \hline
 $g$ & symbolic & inverter 1 and 2 control gain & -- & --\\ \hline
\end{tabular}
\end{center}
\end{table}

\subsection{Conditions for two-LTC System}
\label{2ltc_section}

Let the $\Sigma_1$ be the subsystem where LTC1 and LTC2 operate on the circuit in Fig. \ref{LTC_inv_circuit} normally (see Assumption \ref{assump:ov_uv}). Both devices have the same deadband width $2\varepsilon$ that is centered on the same voltage reference $v^{ref}$. When the voltage is outside the deadband for $d_{L1}$ ($d_{L2})$ seconds, LTC1 (LTC2) taps, which updates both voltages according to $v[k+1]=v[k]\pm [\bar{v}_L ~\bar{v}_L]^\top$ where $v=[v_1 ~v_2]^\top \in \mathbb{R}^2$. Because all tap actions shift both voltages by $\bar{v}_L$ amount, tapping manifests as discrete jumps on the $(v_1,v_2)$ space with slope of $\pm 1$ between the initial and after-tap voltage. 

\begin{lemma}
    If $\bar{v}_L>2\varepsilon$, system $\Sigma_1$ will have marginally stable oscillations for all time when any $v_1[0] \in M(D,c)$ or $v_2[0] \in M(D,c)$ where $c=\bar{v}_L-2\varepsilon>0$. \label{lem_obv_2LTC}
\end{lemma}

Distribution engineers know not to set $\bar{v}_L>2\varepsilon$ when choosing LTC settings, so next we focus on how oscillations could occur when $\bar{v}_L\leq 2\varepsilon$. 

We partition $W$ into four regions $D$, $W_g$, $W_b$, $W_o$, based on the possible trajectories from starting the system in each region. We define $W_o$ such that from there we only transition to the deadband or oscillate. For example, $v_1$ should satisfy $v_1-\bar{v}_L<v^+$. Therefore we define the boundary of $W_o$ in terms of the $v_i^*$ that satisfies $v_i^*-\bar{v}_L= v^+ ~\text{for}~ i=1,2$. This gives the regions 
\begin{subequations}
\begin{align}
\begin{split}
    W_o &= \{(v_1,v_2) \in W: (v_1 \in M(D,v_1^*-v^+) ~\text{and}~v_2\in D),\\
      &\qquad ~\text{or}~(v_2\in M(D,v_2^*-v^+) ~\text{and}~v_1\in D)\},
\end{split}\\
\begin{split}
    W_b &= \{(v_1,v_2) \in W: (v_1,v_2>v^+ ~\text{and}~ v_1+v_2\leq v^+ \\&+\bar{v}_L),
      ~\text{or}~  (v_1,v_2<v^- ~\text{and}~ v_1+v_2\geq v^- -\bar{v}_L)\},
\end{split}
\end{align}
\label{2ltc_region_defns}
\end{subequations}
and $W_g$ is what remains of $W$ ($W_g=W \setminus (D \cup W_o \cup W_g$)).
\begin{figure}[!h]
    \centering
    \begin{subfigure}{0.4\columnwidth}
        \begin{center}
	    \resizebox{0.99\columnwidth}{!}{%
		\begin{tikzpicture}
			\input{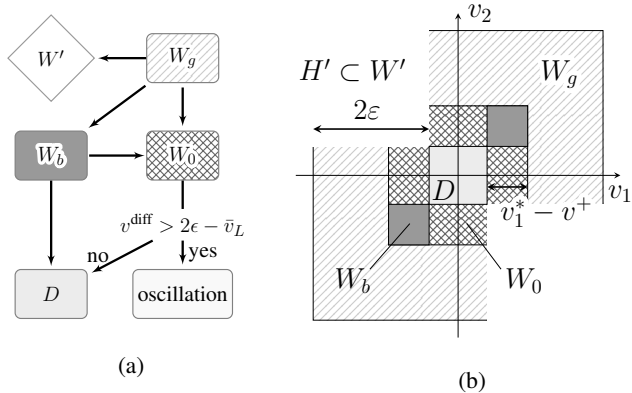}
		\end{tikzpicture}
		}
	    \end{center}
        \caption{}
    \end{subfigure}
    \hfill
    \begin{subfigure}{0.55\columnwidth}
        \begin{center}
	    \resizebox{0.99\columnwidth}{!}{%
		\begin{tikzpicture}
			\Large
\begin{axis}[
unit vector ratio*=1 1 1,,
 	xmin=-\axlim,
 	xmax=\axlim,
 	ymin=-\axlim,
 	ymax=\axlim,
 	axis lines=none]

	\draw[pattern={north east lines},pattern color=green] (-\epsi-3*\VL,-\epsi-3*\VL) rectangle (\epsi+3*\VL,\epsi+3*\VL) node[pos=.85] {\contour{white}{$W_g$}};
	\draw[thin, pattern={crosshatch},pattern color=orange] (-\epsi-0.01,-\epsi-0.01) rectangle (\epsi+0.01,\epsi+0.01);
	\draw[thin, fill=blue] (-\epsi-0.01,-\epsi-0.01) rectangle (-0.5*\epsi,-0.5*\epsi);
	\draw[thin, fill=blue] (0.5*\epsi,0.5*\epsi) rectangle (\epsi+0.01,\epsi+0.01);
	\draw[draw=none, fill=white] (-0.5*\epsi,0.5*\epsi) rectangle (-\epsi-3*\VL-0.01,\epsi+3*\VL+0.01) node[pos=.6] {$H' \subset W'$};
	\draw[draw=none, fill = white] (0.5*\epsi,-0.5*\epsi) rectangle (\epsi+3*\VL+0.01,-\epsi-3*\VL-0.01);
	\draw[thin, fill=gray!15] (-0.5*\epsi,-0.5*\epsi) rectangle (0.5*\epsi,0.5*\epsi) node[pos=.25] {$D$};
\end{axis}

\begin{axis}[
unit vector ratio*=1 1 1,,
 	axis x line=center,
 	axis y line=center,
 	xlabel={$v_1$},
 	ylabel={$v_2$},
 	xlabel style={below},
 	ylabel style={right},
 	xmin=-\axlim,
 	xmax=\axlim,
 	ymin=-\axlim,
 	ymax=\axlim,
 	ticks=none,
    ]

	\draw[thick, >=stealth, <->, ] (-\epsi-3*\VL,\axlim - 0.1) -- (-0.5*\epsi,\axlim - 0.1) node[pos=.5, above] {};
	\draw[thick, >=stealth, <->, ] (0.025,-0.01) -- (0.06,-0.01);
\end{axis}
\draw node at (1.25,4) {$2 \varepsilon$};
\draw node at (4,1) {$W_0$};
\draw node at (1,1) {$W_b$};
\draw node at (4.4,2.2) {\contour{white}{$v_1^* - v^+$}};
\draw[thin] (3.6,1.2) -- (3,2);
\draw[thin] (1.3,1.2) -- (2,2);
		\end{tikzpicture}
		}
	    \end{center}
        \caption{}
    \end{subfigure}
    \caption{Regions for two-LTC system $\Sigma_1$ when $\bar{v}_L<2\varepsilon$.}
    \label{2ltc_regions} %
\end{figure}

Fig. \ref{2ltc_regions} shows these regions in the state space as well as the possible transitions between regions $W_g, W_o, W_b, D$, and $W^\prime$. 
For $W_g$ possibilities, LTC jumps having a slope of $\pm 1$ implies that the only way $W_g \rightarrow W'$ is by $W_g \rightarrow H' \subset W'$. Because $W_g$ is $\bar{v}_L$ away from $D$, $W_g \nrightarrow D$. %
For $W_b$ possibilities, LTC jumps having a slope of $\pm 1$ implies that $W_b$ cannot transition to $W^\prime$ nor $W_g$. %
For $W_o$ possibilities, LTC jumps having a slope of $\pm 1$ implies that the only way $W_o \rightarrow W'$ is by $W_o \rightarrow H' \subset W'$. By the same logic, $W_o$ cannot transition to $W_b$ nor $W_g$. %
Finally, from the geometry of Fig. \ref{2ltc_regions}, observe that if $v^{\text{diff}}[0]<2\varepsilon$, $W_o \nrightarrow W^\prime$.

\begin{lemma} \label{lem_2LTC}
    If $\bar{v}_L\leq 2 \varepsilon$, $v[T] \in W_o$, and $v^{\text{diff}}<2\varepsilon-\bar{v}_L$, system $\Sigma_1$ will have marginally stable oscillations starting at time $T$.
\end{lemma}

A simulation of marginally stable oscillations due to Lemma \ref{lem_2LTC} hunting is in Fig. \ref{marg_oscill_stap}.
\begin{figure}[!h]  
   \centering 
   \includegraphics[width=\columnwidth]{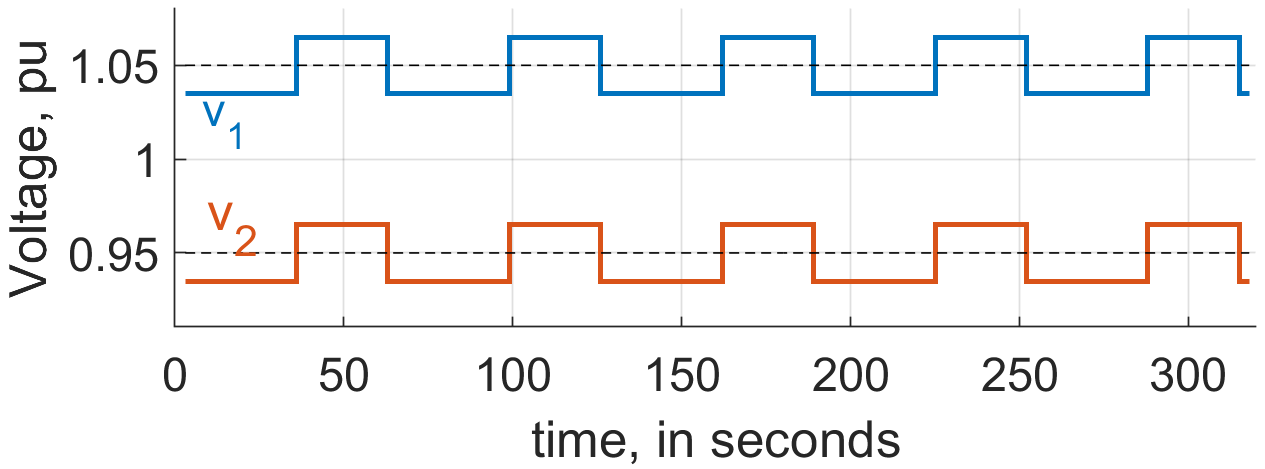}
    \caption{Simulation of system $\Sigma_1$ when Lemma \ref{lem_2LTC} holds. Parameters are the defaults listed in Table \ref{tab:base_parms}.}
    \label{marg_oscill_stap} %
\end{figure}
\subsection{Conditions for 2-Inverter System}
Let the $\Sigma_2$ be the subsystem where the two inverters operate on the circuit in Fig. \ref{LTC_inv_circuit} normally (see Assumption \ref{assump:ov_uv}).
For now we omit the deadband in the control logic. We employ a discrete integrator (also called incremental volt-var control in \cite{Farivar_incremental,Guannan}) for computing inverter reactive power set-points $q^{inv}$ with
\begin{equation} 
    q^{inv}[k+1]=q^{inv}[k]-G(v[k]-v^{ref})\label{ulaw}, %
\end{equation}
where $v^{ref}\in \mathbb{R}^{n}$ is the reference voltage, and $G$ is a diagonal matrix containing controller gains. Because only a subset of the network nodes are controlled, $v^{ref}$ can be assigned to the nominal vector of ones. $G$ being diagonal enforces that each inverter is injecting power to regulate the voltage at its own node. For any radial circuit with $n$ inverters at different nodes, $v \in \mathbb{R}^{n}$ and $q^{inv} \in \mathbb{R}^{n}$.

The algebraic power flow equations that map inverter power injections to voltages can be represented by
\begin{equation}
   v[k+1]=v[k]+X(q^{inv}[k+1]-q^{inv}[k]) \label{Veqn} %
\end{equation}
from \cite[equation 8]{Helou}. Matrix element $X_{ij}$ is the common ancestor path reactance between node $i$ and node $j$ on the network. Next we substitute \eqref{ulaw} into \eqref{Veqn}, giving
\begin{equation}
   v[k+1]=v[k]-X G(v[k]-v^{ref}) \label{v_update}.
\end{equation}
Finally, we subtract $v^{ref}$ from both sides and define $e[k] \coloneqq v[k]-v^{ref}$, giving \cite[equation 12]{Helou}:
\begin{equation}
   e[k+1]=(I-XG)e[k]. \label{QV_int_closedloop} %
\end{equation}

Author \cite{Helou} proves in their Theorem 3.1 that for the system \eqref{QV_int_closedloop}, $v \rightarrow v^{ref}$ iff
\begin{equation}
    0 \prec G \prec 2X^{-1}. \label{helou_converge}
\end{equation}

In the scalar case (one inverter acting on a single phase circuit), this condition is $0<g<2/\chi$. This implies that under normal operation, $g$ should be positive but not too large to have the voltages converge. We are also interested in the possibility of $g<0$, where the inverters push the voltage away from $v^{ref}$, because that case is more dangerous.
 
\begin{lemma}
    If $G \prec0$, system $\Sigma_2$ given by \eqref{QV_int_closedloop} has $v \rightarrow \pm \infty$. \label{QV_int_diverge}
\end{lemma}
When the deadband is introduced, inverters only operate according to \eqref{ulaw} when their voltage is outside the deadband. Because $v^{ref}$ is centered in the deadband, Lemma \ref{QV_int_diverge} holds in the same way as the no-deadband case, and the convergence condition \eqref{helou_converge} yields $v \rightarrow D$ instead of $v \rightarrow v^{ref}$.

 For our system $\Sigma_2$, $G=diag([g ~g])$ and $q=[q_1~ q_2]$. The dissipative nature of power grids due to line impedances causes the diagonal terms of $X$ to be larger than the off-diagonal terms. From the circuit in Fig. \ref{LTC_inv_circuit}, system $\Sigma_2$ has
\begin{align}
    X=\begin{bmatrix}
        \chi & \eta\chi\\
        \eta\chi & \chi
    \end{bmatrix}, \label{damp_eqn}
\end{align}
where the damping factor $\eta$ satisfies $0<\eta<1$. In general, if inverters have different reactances in the line path to the substation, $\eta_1=X_{21}/X_{11}$ and $\eta_2=X_{12}/X_{22}$. Because in this work we only use $\eta$ for its property of $0<\eta<1$, using $\eta=\eta_1=\eta_2$ does not change the results. 
\section{Full System Analysis}
\label{probform2_section}

\subsection{Modeling the Four Devices as a Hybrid System}
\label{hybrid_sys_section}
 Next we model all devices in Fig. \ref{LTC_inv_circuit} operating normally (see Assumption \ref{assump:ov_uv}) as a discrete hybrid automaton, which is the interconnection of a finite state machine with a switched affine system.
 This system, $\Sigma^3$, has state vector $x=\begin{bmatrix}z_1 & z_2 & z_3 & v_1 & v_2\end{bmatrix}^\top$, where $z_1,z_2,z_3$ are the internal timers for LTC1, LTC2, and inverter1 and 2, respectively. The two inverters use the same timer $z_3$ because they have the same delay of $d_{inv}$. Let $T_s$ be the timestep of the discrete model. Each mode has a label with the format $m^{\star}0$ where $\star=1,2,...8$, and has affine dynamics of the form $x[k+1]=A x[k]+c$ where $A \in \mathbb{R}^{5 \times 5}$ and $c \in \mathbb{R}^{5}$. 
To define the switching conditions, we define a function $f$ for whether a voltage is inside the deadband:
\begin{equation}
    f_i \coloneqq max(v_i-v^+,0)-max(v^{-}-v_i,0) ~\text{for}~i=1,2. %
\end{equation} 
 For $i=1,2$, if $v_i$ is an overvoltage then $f_i>0$. If it is an undervoltage then $f_i<0$, and if inside $D$ then $f_i=0$. Inverter1 (LTC1) responds when $f_1 \neq 0$ for $d_{inv}$ ($d_{L1}$) seconds, and Inverter2 (LTC2) responds when $f_2 \neq 0$ for $d_{inv}$ ($d_{L2}$) seconds.

Now we introduce the hybrid model, where we make the symbolic variables bold:
\begin{align}
    \shortintertext{m10: tap LTC1 up} %
    \shortintertext{Switch condition: $~z_1>\mathbf{d_{L1}} ~\text{and}~ (f_1+f_2<0)$}
    \shortintertext{Dynamics:}
     &   \begin{bmatrix} z_1 \\z_2 \\z_3 \\ v_1 \\ v_2 \end{bmatrix}_{k+1} = \begin{bmatrix}
            0 & 0 & 0 & 0 & 0\\
            0 & 1 & 0 & 0 & 0\\
            0 & 0 & 1 & 0 & 0\\
            0 & 0 & 0 & 1 & 0\\
            0 & 0 & 0 & 0 & 1\\
        \end{bmatrix}
        \begin{bmatrix} z_1 \\z_2 \\z_3 \\ v_1 \\ v_2 \end{bmatrix}_{k}+\begin{bmatrix} T_s \\ T_s \\ T_s \\ \mathbf{ \bar{v}_L} \\ \mathbf{\bar{v}_L} \end{bmatrix}\nonumber
    \\\cline{1-2}
    \shortintertext{m20 (tap LTC1 down) is the same as m10 except the condition has $(f_1+f_2)>0$ and the affine term is $[T_s~T_s~T_s ~ \mathbf{-\bar{v}_L} ~ \mathbf{-\bar{v}_L}]^\top$.}  
    \cline{1-2}
    \shortintertext{m30: tap LTC2 up} %
    \shortintertext{Switch condition: $~z_2>\mathbf{d_{L2}} ~\text{and}~ (f_1+f_2<0)$}
    \shortintertext{Dynamics:}
     &   \begin{bmatrix} z_1 \\z_2 \\z_3 \\ v_1 \\ v_2 \end{bmatrix}_{k+1} = \begin{bmatrix}
            1 & 0 & 0 & 0 & 0\\
            0 & 0 & 0 & 0 & 0\\
            0 & 0 & 1 & 0 & 0\\
            0 & 0 & 0 & 1 & 0\\
            0 & 0 & 0 & 0 & 1\\
        \end{bmatrix}
        \begin{bmatrix} z_1 \\z_2 \\z_3 \\ v_1 \\ v_2 \end{bmatrix}_{k}+\begin{bmatrix}T_s \\ T_s \\ T_s \\ \mathbf{ \bar{v}_L} \\ \mathbf{\bar{v}_L} \end{bmatrix} \nonumber
        \\\cline{1-2}
    \shortintertext{m40 (tap LTC2 down) is the same as m30 except the condition has $(f_1+f_2)>0$ and the affine term is $[T_s~T_s~T_s ~ \mathbf{-\bar{v}_L} ~ \mathbf{-\bar{v}_L}]^\top$.} 
    \cline{1-2}
    \shortintertext{m50: inverter(s) respond to voltage issues}
    \shortintertext{Switch condition: $~z_3>\mathbf{d_{inv}}~\text{and}~ (f_1 \neq0 ~\text{or}~ f_2 \neq0$)}
    \shortintertext{Dynamics:}
     &\begin{bmatrix}
              z_1\\z_2\\z_3\\ v_1 \\ v_2
        \end{bmatrix}_{k+1}=
        \xbracketMatrixstack{
            1 & 0 & 0 & 0 & 0\\
            0 & 1 & 0 & 0 & 0\\
            0 & 0 & 0 & 0 & 0\\
            0& 0& 0 & 1-\chi\mathbf{g} & -\eta\chi\mathbf{g}\\
             0& 0& 0& -\eta\chi\mathbf{g} & 1-\chi\mathbf{g}
        }
        \begin{bmatrix}
            z_1\\z_2\\z_3\\ v_1 \\ v_2
        \end{bmatrix}_k+\nonumber
    \shortintertext{$\begin{bmatrix}
            T_s&T_s&T_s& \chi\mathbf{g}v^{ref}+s\chi\mathbf{g}v^{ref} & s\chi\mathbf{g}v^{ref}+\chi\mathbf{g}v^{ref}
        \end{bmatrix}^\top$ \label{eq:mode60}}
    \cline{1-2}
    \shortintertext{m60: reset LTC2 and inverter timers}
    \shortintertext{Switch condition: [$m=m10$ and $(f_1>0$ or $f_2>0)$] $\quad\quad$ or [$m=m20$ and $(f_1<0$ or $f_2<0)$]}
    \shortintertext{Dynamics: $z_2[k+1]=T_s,~ z_3[k+1]=T_s, \quad v[k+1]=v[k]$}\cline{1-2}
     \shortintertext{m70: reset LTC1 and inverter timers}
    \shortintertext{Switch condition: [$m=m30$ and $(f_1>0$ or $f_2>0)$] $\quad\quad$ or [$m=m40$ and $(f_1<0$ or $f_2<0)$]}
    \shortintertext{Dynamics: $z_1[k+1]=T_s,~ z_3[k+1]=T_s, \quad v[k+1]=v[k]$}\cline{1-2}
    \shortintertext{m80: increment timers}
    \shortintertext{Switch condition: no other mode conditions hold}
    \shortintertext{Dynamics: $z_i[k+1]=z_i[k]+T_s ~\forall i=1,2,3, ~v[k+1]=v[k]$} \nonumber
\end{align}
\begin{remark} \label{remark:delay}
By the relationship $d_{inv} < d_{L1} < 2 d_{L2}$ from Table \ref{tab:base_parms}, the inverters respond to voltage issues before LTC1. This relationship is a less conservative version of $d_{inv} \ll d_{L1}$, which is often made for power systems \cite{Ahmed}. Therefore, an LTC will never tap twice before an inverter acts. 
\end{remark}
\subsection{Behavior when $g>0$}
For positive values of the inverter gain $g$, our control action renders the voltage dynamics stable, as shown next.  

\begin{proposition}
    If system $\Sigma_3$ has $0< g < \frac{2}{\chi}$, there exists $T>0$ such that $v[k] \in D$ for all $k \geq T$. \label{lem_gpos}
\end{proposition}

\begin{proof}

    We give an intuitive sketch of the proof through a Lyapunov argument. We notice that when $g=0$ circuit $\Sigma_3$ behaves equivalently to $\Sigma_1$. Lemma \ref{lem_2LTC} establishes that if $\bar{v}_L<2 \varepsilon$, then the system exhibits marginally stable oscillations (recall Fig. \ref{marg_oscill_stap}). This implies that the energy of the system 
         \newpage %

    $Y_k$ remains positive and bounded by fixed values $\bar Y_1$ and $\bar Y_2$, i.e., $0<\bar Y_1 \leq Y_k \leq \bar Y_2 < \infty$, when $g = 0$.

    Now if we consider a positive gain $0<g<2/\chi$ in $\Sigma_3$, equation \eqref{helou_converge} guarantees that the inverter acts as a stabilizing controller during mode m50, decreasing the system energy and driving $v$ closer to the deadband.
    Finally, from Remark \ref{remark:delay} we know that because of the conditions on the timers, 
    we will always return to m50 regularly, at least once between each LTC tap.
    Therefore, the system $\Sigma_3$ has maximum energy of $\bar Y_2$ and will only lose energy until landing in the deadband. 
\end{proof}

The consequence of Lemma \ref{lem_gpos} is that for realistic values of the network with $g >0$, the system exhibits damped oscillations which eventually reach the deadband. Thus, we focus the remainder of this paper on the more dangerous scenario where $g<0$ renders the oscillations unstable. 
\subsection{Trajectory Walkthrough for $g<0$}
\begin{figure}[!h]  
   \centering 
   \includegraphics[width=.5\textwidth]{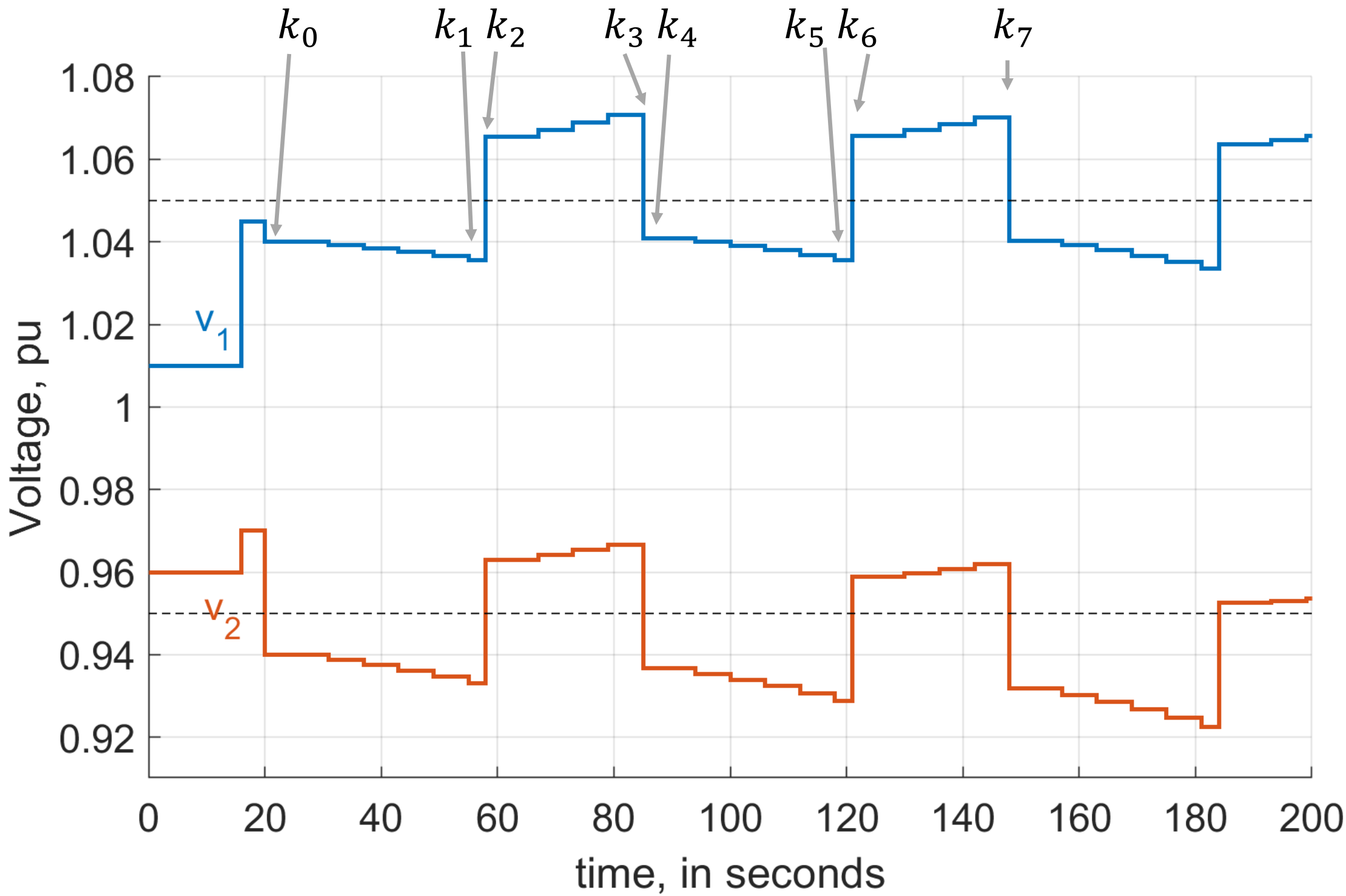}
    \caption{Simulation of system $\Sigma_3$ for $g<0$ with parameters being the defaults in Table \ref{tab:base_parms}. System start time is at $k_0$.}
    \label{fig_walkthrough} 
\end{figure}

Here, we describe a scenario on system $\Sigma_3$ that can result in the beginning of unstable oscillations. Suppose an adversary gains access to and installs malware in the communication system that remotely sends parameters to the inverters. The malware logic could, for example, send a negated inverter controller gain $g$ when any of the inverter voltages suddenly change substantially ($|v[k+1]-v[k]|>0.6 \varepsilon$), and the voltage lands far from $v^{ref}$ ($|v[k+1]-v^{ref}|>0.6 \varepsilon$). Such an attack is dangerous because the trigger by external voltage event(s) conceals the time when the cyber \emph{breach} occurred, which could be much earlier than the parameter negation.

Follow along with Fig. \ref{fig_walkthrough}. Suppose shortly before time $k_0$ some external event(s), such as a fault or effect from the transmission grid, causes the voltages to shift abruptly. Suppose the voltage shifts increase $v^{\text{diff}}$, make $v_2$ an undervoltage, and triggers the negation of the inverter parameter $g$. See \cite{Birks_repo} for power flow examples of this voltage shift on realistic circuits. We are interested in what happens after the system $\Sigma_3$ start time of $k_0$. From $k_0$ to $k_1$: at $k_0$ we have $f_2<0$ while $f_1=0$, so 
inverter2 actuates several times, each time causing $v_2$ to get further from the deadband. For $k_1$ to $k_2$: the LTC2 delay is complete so LTC2 taps, fixing the $v_2$ undervoltage but making $v_1$ an overvoltage. For $k_2$ to $k_3$: 
the inverter1 actuates several times, each time causing $v_1$ to get further from the deadband. For $k_3$ to $k_4$: the LTC1 delay is complete so LTC1 taps, fixing the $v_1$ overvoltage but making $v_2$ an undervoltage again. At $k_4$ we have $f_2<0$ while $f_1=0$ which was the case for $k_0$, so we have completed one quasi-periodic oscillation. In future sections we show how these oscillations continue. 

\subsection{Single Inverter Action Preliminaries}
In this section we present some properties of inverter actions between LTC taps to prepare for later proofs. %

 Suppose inverter $i$ is acting during a time interval $[k ~...~ k+N]$. Define the change in $v_j \in \mathbb{R}^1$ due to the inverter's action as $\Delta v^{inv}_j(X_{ij},k+N,k,v_i[k]) \coloneqq v_j[k+N]-v_j[k] \in \mathbb{R}^1$. $\Delta v_j^{inv}$ can be thought of as the projection of inverter's actuation onto the $v_j$ axis of the $(v_i,v_j)$ space. Occasionally we omit some of the four parameters from $\Delta v^{inv}$ when they are not relevant.

\begin{assumption}(Bounds on LTC and inverter action)
    \\$\bar{v}_L+\Delta v^{inv}_i(X_{ii},k+N,k,v_i[k]) < 2\varepsilon ~\text{for}~ v_i[k] \in W\text{, and} ~\forall k,N$. %
    \label{steepest}
\end{assumption}
This assumption implies that $\bar{v}_L<2\varepsilon$ and $\Delta v^{inv}_i(X_{ii},k+N,k,v_i[k]) < 2\varepsilon ~\text{for}~ v_i[k] \in W\text{, and} ~\forall k,N$. The Assumption is reasonable because $\bar{v}_L$ is typically significantly less than $\varepsilon$, and $\Delta v^{inv}$ imparting a voltage change of close to $2\varepsilon$ would require unreasonably large combinations of circuit impedance and inverter capacity.

Recall the voltage update equation \eqref{v_update} $v[k+1]=v[k]-X G(v[k]-v^{ref})$. 
The inverter acts $N_1 \coloneqq \text{floor}(d_{L1}/d_{inv})$ times if after the interval $[k ~...~ k+N_1]$ the LTC1 taps, or $N_2 \coloneqq \text{floor}(d_{L2}/d_{inv})$ times if after $[k ~...~ k+N_2]$ LTC2 taps.
The $i^{th}$ row of \eqref{v_update} implies that node $i$'s voltage is
\begin{equation}
        v[k+1]=v[k]+\chi(v[k]-v^{ref}) ~\in \mathbb{R}^1, \label{v_update_scalar1}
\end{equation}
and each other voltage on the network is given by
\begin{equation}
       v[k+1]=v[k]+\eta\chi(v[k]-v^{ref}) ~\in \mathbb{R}^1 \label{v_update_scalar2}
\end{equation}
from substituting \eqref{damp_eqn} into \eqref{v_update}. Now consider the accumulation of \eqref{v_update} for $N=N_1$ or $N=N_2$ timesteps:
\begin{equation}
    v[N+k]=(I-X G)^N v[k]+\sum_{r=0}^{N-1} (1-XG)^r(X G v^{ref}). \label{affine_soln1}
  \end{equation} 
The $j^{th}$ row of \eqref{affine_soln1} gives a parameterization of $\Delta v_j^{inv}$:
\begin{multline}
    \Delta v^{inv}_j(X_{ij},k+N,k,v_i[k],g) = v_j[N+k]-v_j[k] \\=((1-\chi g)^N-1)v[k]+\sum_{r=0}^{N-1} (1-\chi g)^r(X_{ij} g v^{ref}). \label{parametrized}
\end{multline}

\begin{remark}(Remarks about $\Delta v^{inv}_i ~\forall~ i=1,2)$ \label{remark_deltaV_sign}
\begin{enumerate}
    \item Because $\Delta v_i^{inv}<2\varepsilon$ from Assumption \ref{steepest}, $\Delta v_i^{inv}(X_{ii},k+1,k)$ has the same sign for $k=0...N$.
   \item If $g>0$, $sign \{ \Delta v_i^{inv}(v[k])\} = -sign \{ v_i[k]-v^{ref}\}$. %
   \item If $g<0$, $sign \{ \Delta v_i^{inv}(v[k])\} = sign \{ v_i[k]-v^{ref}\}$. %
\end{enumerate}
\end{remark}
The first item establishes that each inverter actuates in the same direction between LTC taps. The second (third) items establish that when $g>0$ ($g<0$), the inverters push the voltage toward (away) from the deadband.

\begin{lemma} ($\Delta v^{inv}$ is a homogeneous function)
    The coupling effect of a single inverter actuating at node $i$ on the voltage at node $j$ is damped by a factor of $\eta$. That is, 
    $\Delta v_j^{inv}(X_{ij},k+N,k,v_i[k])=\eta \Delta v_i^{inv}(X_{ii},k+N,k,v_i[k])$.
    \label{homogeneous}
\end{lemma}

\subsection{Partitioning $W$ in State Space when $g<0$}
\label{4-device_region_section}
In this section we will partition $W$ into $D$, $W_g$, $W_b$, $W_o$,  based on the possible trajectories from starting the system in each region. We use a similar process to Section \ref{2ltc_section}. 
The $W_b$ region will be when both voltages are above or below the deadband. In that case, both inverters act between LTC taps, and from Assumption \ref{steepest} oscillations cannot occur. Then $W_o$ will be where only one inverter acts between taps, and is close enough to $D$ for oscillations to occur. An oscillation would begin with an inverter pushing the voltages further from the deadband (Remark \ref{remark_deltaV_sign} \#3) until an LTC tap towards the deadband and overshoots it. Therefore, the boundary of $W_o$ comes from the states $v$ where after inverter action(s) and an LTC tap the state is within the deadband edge:
\begin{subequations}
\begin{align}
    v_1+r_1-\bar{v}_L&<v^+ \label{r_eqn1}\\
    v_2+r_2-\bar{v}_L&<v^+ \label{r_eqn2}\\
    v_1+r_1+\bar{v}_L&>v^- \label{r_eqn3}\\
    v_2+r_2+\bar{v}_L&>v^-. \label{r_eqn4}
\end{align}\label{r_eqns}
\end{subequations}
Function $r$ is the inverters' change in $v$ before an LTC taps:
\begin{subequations}
    \begin{align}
        r_1 &\coloneqq \Delta v^{inv}_1(\chi,N_1,v_1)+\Delta v^{inv}_1(X_{12},N_1,v_2) \label{r_defn1}\\
        r_2 &\coloneqq \Delta v^{inv}_2(X_{12},N_1,v_1)+\Delta v^{inv}_2(\chi,N_1,v_2). \label{r_defn2}
    \end{align}
\end{subequations}
Note that $r$ depends on both $v_1$ and $v_2$, but when only one inverter acts, one term in each of \eqref{r_defn1} and \eqref{r_defn2} zeros out, causing \cref{r_eqn1,r_eqn2,r_eqn3,r_eqn4} to depend on only $v_1$ \emph{or} $v_2$.
We can now define the state space regions as
\begin{subequations}
\begin{align}
\begin{split}
    W_o = &\{(v_1,v_2) \in W:  \\&(v_1 \in M(D,v_1^*-v^+) ~\text{and}~v_2\in D),\\
      &\text{or}~(v_2\in M(D,v_2^*-v^+) ~\text{and}~v_1\in D)\},
\end{split}\label{Wo_4device}\\
\begin{split}
    W_b = &\{(v_1,v_2) \in W:  \\&(v_1,v_2>v^+, v_2\leq v_1, \eqref{r_eqn1} ~\text{holds}),\\
    &\text{or}~ (v_1,v_2>v^+, v_2> v_1, \eqref{r_eqn2} ~\text{holds}),\\
    &\text{or}~ (v_1,v_2<v^-, v_2> v_1\eqref{r_eqn3} ~\text{holds}),\\
    &\text{or}~ (v_1,v_2<v^-, v_2\leq v_1, \eqref{r_eqn4} ~\text{holds})\},
\end{split}\label{Wb_4device}
\end{align}
\label{2ltc_region_defns}
\end{subequations}
and $W_g$ is what remains of $W$ ($W_g=W \setminus (D \cup W_o \cup W_g$).
 The $v_1^*$ and $v_2^*$ in \eqref{Wo_4device} are the $v_1$ and $v_2$ when \eqref{r_eqn1} and \eqref{r_eqn2} are set to be equalities. The $W_o$ region definition is the same as that of the 2-LTC system $\Sigma_1$ \eqref{2ltc_region_defns}, with the distinction that equations \eqref{r_eqns} have the additional $r$ term. 
\begin{figure}[!h]
    \centering
    \begin{subfigure}{0.4\columnwidth}
        \begin{center}
	    \resizebox{0.99\columnwidth}{!}{%
		\begin{tikzpicture}
			\input{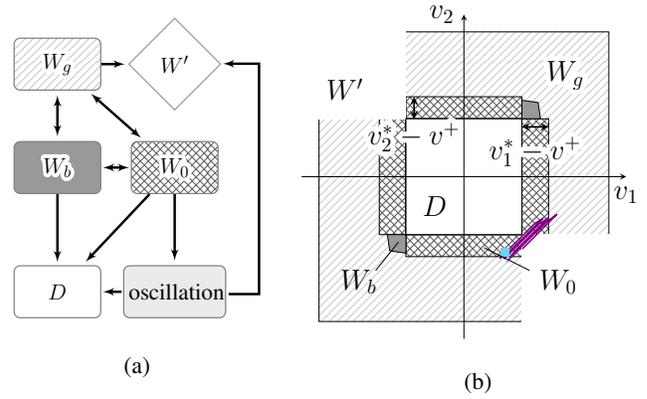}
		\end{tikzpicture}
		}
	    \end{center}
        \caption{}
    \end{subfigure}
    \hfill
    \begin{subfigure}{0.55\columnwidth}
        \begin{center}
	    \resizebox{0.99\columnwidth}{!}{%
		\begin{tikzpicture}
			\Large
\begin{axis}[
unit vector ratio*=1 1 1,,
 	xmin=-\axlim,
 	xmax=\axlim,
 	ymin=-\axlim,
 	ymax=\axlim,
 	axis lines=none]

	\draw[pattern={north east lines},pattern color=green] (-\epsi-3*\VL,-\epsi-3*\VL) rectangle (\epsi+3*\VL,\epsi+3*\VL) node[pos=.85] {\contour{white}{$W_g$}};
	\draw[thin, pattern={crosshatch},pattern color=orange] (-\epsi-\VL+\pr,-\epsi) rectangle (\epsi,\epsi); %
	\draw[thin, pattern={crosshatch},pattern color=orange] (\epsi+\VL-\pr,\epsi) rectangle (-\epsi,-\epsi); %
	\draw[thin, pattern={crosshatch},pattern color=orange] (\epsi,\epsi+\VL-\prr) rectangle (-\epsi,-\epsi); %
	\draw[thin, pattern={crosshatch},pattern color=orange] (-\epsi,-\epsi-\VL+\prr) rectangle (\epsi,\epsi); %
	\draw[thin, fill=blue] (-\epsi-\VL+\pr+0.009,-\epsi-\VL+\pr+0.009) -- (-\epsi,-\epsi-\VL+\pr+0.007) -- (-\epsi,-\epsi) -- (-\epsi-\VL+\pr+0.007,-\epsi) -- (-\epsi-\VL+\pr+0.009,-\epsi-\VL+\pr+0.009);
	\draw[thin, fill=blue] (\epsi+\VL-\pr-0.009,\epsi+\VL-\pr-0.009) -- (\epsi,\epsi+\VL-\pr-0.007) -- (\epsi,\epsi) --  (\epsi+\VL-\pr-0.007,\epsi) -- (\epsi+\VL-\pr-0.009,\epsi+\VL-\pr-0.009);
	\draw[draw=none, fill=white] (-\epsi,\epsi) rectangle (-\epsi-3*\VL-0.1,\epsi+3*\VL+0.01) node[pos=.3] {$W'$};
	\draw[draw=none, fill=white] (\epsi,-\epsi) rectangle (\epsi+3*\VL+0.01,-\epsi-3*\VL-0.01);
	\draw[thin, fill=white] (-\epsi,-\epsi) rectangle (\epsi,\epsi) node[pos=.25] {$D$};
\end{axis}

\begin{axis}[
unit vector ratio*=1 1 1,,
 	axis x line=center,
 	axis y line=center,
 	xlabel={$v_1$},
 	ylabel={$v_2$},
 	xlabel style={below},
 	ylabel style={left},
 	xmin=-\axlim,
 	xmax=\axlim,
 	ymin=-\axlim,
 	ymax=\axlim,
 	ticks=none,
    ]

	\addplot[color = violet, thick] table [col sep=comma] {trajectory.csv}; 
	\node[rectangle,fill=lightGreen,inner sep=0.07cm] at (0.035,-0.065) {};
	\draw[thick, >=stealth, <->, ] (\epsi+\VL-\pr,\epsi+\VL-\pr - 0.03) -- (\epsi,\epsi+\VL-\pr - 0.03) node[pos=.5, below] {\contour{white}{$v_1^*-v^+$}};
	\draw[thick, >=stealth, <->, ] (-\epsi-\VL+\pr + 0.03, +\epsi+\VL-\prr) -- (-\epsi-\VL+\pr + 0.03, +\epsi,) node[pos=.7, below] {\contour{white}{$v_2^*-v^+$}};
\end{axis}

\draw node at (4.5,1) {$W_0$};
\draw node at (1,1) {$W_b$};
\draw[thin] (4,1.2) -- (3.2,1.6);
\draw[thin] (1.3,1.2) -- (1.7,1.7);
		\end{tikzpicture}
		}
	    \end{center}
        \caption{}
    \end{subfigure}
    \caption{Regions for four-device system $\Sigma_3$ when $g<0$.}
    \label{4device_regions}
\end{figure}

Fig. \ref{4device_regions} shows these state space regions as well as the possible transitions between regions. The trajectory of Fig. \ref{fig_walkthrough} is plotted on Fig. \ref{4device_regions} with a blue square marker for the IC. Observe that the trajectory oscillates until eventually leaving $W$. In addition to the possible transitions for system $\Sigma_1$ in Fig. \ref{2ltc_regions}, this system has: $W_b \rightarrow W_g$, $W_o \rightarrow W_b$, and $W_o \rightarrow W_b$ since the inverter pushes voltages away from $D$. 

When the IC is in $W_g$ or $W_b$, there may exist a $k_0>0$ where the full state vector $x[k_0]=[0,0,0,v_1[0],v_2[0]]^\top$. Because $\Sigma_3$ is time-invariant, the Fig. \ref{4device_regions} regions apply to any $x[k_0]$ where $x[k_0]=x[0]$. For example, the system IC could be in $W_g$ then later have $x[k_0] \in W_o$, after which point the behavior would be the same as if the system started in $W_o$. 
\subsection{Conditions for Oscillations to Begin when $g<0$}

In this section we consider trajectories where $v[k_0] \in W_o$. As illustrated in Fig. \ref{4device_regions}, $W_o$ is comprised of four disjoint regions. If hunting occurs in the lower region of $W_o$ ($f_1=0$ while $f_2<0$), the hybrid system mode sequence (MS) that creates one period of oscillation is $\alpha_1\coloneqq \{m50,m10,m60,m50,m40,m70\}$ after omitting the increment mode ($m80$) for brevity. Similarly, let $\alpha_2$, $\alpha_3$, and $\alpha_4$ be the oscillation sequences when the IC is in the upper, left, and right-hand regions of $W_o$, respectively.

\begin{lemma}(Basis step for oscillations)
    Consider system $\Sigma_3$ with $g<0$. When $f_1=0$ while $f_2<0$, necessary and sufficient conditions for completing one oscillation period starting at time $k_0$ are
    \begin{subequations}
     \label{base_cond}
    \begin{align}
        v_1[k_0]+\eta\Delta v_2^{inv}(\chi,k_0+N_2,k_0,v_2[k_0]) > v^- \label{cond1}\\
        v[T_1] > v^- \label{cond2}\\
        v[T_1]+ \eta\Delta v_1^{inv}(\chi,k_0+N_1,k_0,v[T_1]) < v^+ \label{cond3}\\
        v[T_2]+\Delta v^{inv}(\chi,k_0+N_1,k_0,v[T_2])-\bar{v}_L <v^+ \label{cond4}
    \end{align}
    \end{subequations}
    where $v[T_1]= v_2[k_0]+\Delta v_2^{inv}(\chi,k_0+N_2,k_0,v_2[k_0])+\bar{v}_L$, and $v[T_2]=v_1[k_0]+\eta\Delta v_2^{inv}(\chi,k_0+N_2,k_0,v_2[k_0])+\bar{v}_L$. 
    \label{basis_lemma}
\end{lemma}
\begin{proof}
Follow along with Fig. \ref{fig_walkthrough}. We will express sequence $\alpha_1$ in terms of the system voltage trajectories. For the IC, $v_2[k_0]$ is an undervoltage and $v_1[k_0]$ in the deadband. Inverter2 responds to the $v_2$ undervoltage, but due to $g<0$ it decreases both voltages. To prevent inverter2 from pushing both voltages below the deadband, we require that 
\begin{equation}
     v_1[k_0]+\Delta v_1^{inv}(X_{12},k_0+N_2,k_0,v_2[k_0]) > v^-.  \nonumber\\
\end{equation}
Next, with the persisting $v_2$ undervoltage, LTC2 taps which increases both voltages. To create overshoot so that $v_1$ has an overvoltage, we require that
\begin{equation}
    v_2[k_0]+\Delta v_2^{inv}(X_{22},k_0+N_2,k_0,v_2[k_0])+\bar{v}_L > v^-. \nonumber\\
\end{equation}
Inverter1 responds to the $v_1$ overvoltage, but due to $g<0$ it increases both voltages. To prevent inverter1 from pushing both voltages above the deadband, we require that
\begin{equation}
    v[T_1]+ \Delta v_2^{inv}(X_{12},k_0+N_1,k_0,v_1[T_2])
    < v^+. \nonumber\\
\end{equation}
With the persisting $v_1$ overvoltage, LTC1 taps which decreases both voltages. To create overshoot so that $v_2$ becomes an undervoltage, we require that
\begin{equation}
    v[T_2]+\\ \Delta v_1^{inv}(X_{11},k_0+N_1,k_0,v_1[T_2])-\bar{v}_L <v^+.  \nonumber\\
\end{equation}
Finally, we substitute equation \eqref{damp_eqn} into all above equations and apply Lemma \ref{homogeneous} to the first and fourth equation, giving the conditions in the lemma statement.
\end{proof}

\subsection{Showing Oscillations Continue and Grow when $g<0$}
Next we show that once the MS $\alpha_1$ occurs, the next MS will be $\alpha_1$ while $v \in W_o$. The same process can be applied to the other MS $\alpha_2$, $\alpha_3$, and $\alpha_4$.

Define $m[k_i]$ as the mode that system $\Sigma_3$ is in at time $k_i$. Across each $\alpha_1$ MS, $v^{\text{diff}}$ increases. 
We will prove that in the next section, but note that as a result the voltage that is outside $D$ gets further from $v^{ref}$ across each $\alpha_1$ MS. Next, consider the following condition:
\begin{equation}
    v_1^{inv}(k_5,k_4,v_1[k_4])+\Delta v_1^{inv}(k_3,k_2,v_1[k_2])>0 \label{cond5}
\end{equation}

\begin{lemma}(Induction step for oscillations)
    Consider system $\Sigma_3$ with $g<0$. Assuming condition \eqref{cond5} and Assumption \ref{assump:ov_uv} holds, if the system completes MS $\alpha_1$, the MS $\alpha_1$ will repeat.
    \label{inductioN_1neg}
\end{lemma}
\begin{proof}    
Use Fig. \ref{fig_walkthrough} to follow along.
$m[k_4]$: at $k_0$ inverter2 actuates and $m[k_0]=m50$ so $v_2[k_0] \notin D$. Because $v_2$ gets further from $v^{ref}$, $v_2[k_4]<v_2[k_0]$, so $v_2[k_4] \notin D$ too. From Assumption \ref{assump:ov_uv}, $v_1[k_4]$ cannot be outside $D$, so $v_1[k_4] \in D$. Thus inverter2 actuates at $k_4$ and $m[k_4]=m50$.

$m[k_5]$: Because $g<0$, inv2 actuating ($m[k_4]=m50$) will keep $v_2 \notin D$. If the inverter's coupling effect on $v_1$ is strong enough, even though $v_1[k_4] \in D$, $v_1$ could go below the deadband by the time of $k_5$. By Assumption \ref{steepest}, the inverter actions summed with an LTC tap are not large enough for this. Thus at $k_5$, $v_1$ is still inside $D$ and $v_2$ is still outside $D$. So at $k_5$ LTC2 taps up and $m[k_5]=m30$.

$m[k_6]$: The LTC2 tap may not overshoot the deadband at $k_6$. If \eqref{cond5} does not hold, then $v_1[k_6]\in D$, by Assumption \ref{assump:ov_uv} $v_2[k_6] \in D$ too, and the system stays in the deadband. If \eqref{cond5} does hold, inverter1 actuates and $m[k_4]=m50$. 
$m[k_7]$: Assumption \ref{steepest} disallows $v_2$ from going above $D$ (like how in $m[k_5]$ it disallows $v_1$ from going below $D$). Thus at $k_7$, $v_2$ is still inside $D$ and $v_1$ is still outside $D$. So at $k_7$ LCT1 taps down and $m[k_7]=m20$.
\end{proof}
\begin{lemma} (Oscillations grow)
    If system $\Sigma_3$ with $g<0$ has oscillations, $v^{\text{diff}}$ increases after each oscillation period. \label{expansion}
\end{lemma}

Let $s_1$ be the set of $(v_1[k_0],v_2[k_0])$ where basis conditions \eqref{base_cond} hold. The sets of conditions that correspond to sequences $\alpha_2$,$\alpha_3$, and $\alpha_4$ are derived with the same process as the Lemma \ref{basis_lemma} proof, and we call their corresponding sets $s_2$, $s_3$, and $s_4$. Then define $S \coloneqq (\bigcup_{i=1,...,4} s_i) \in W_o$, which is the only voltage region where oscillations can begin. Note that with the $\Delta v^{inv}$ parametrized form \eqref{parametrized}, $S$ can be represented by purely the variables listed in Table \ref{tab:base_parms}.  

 \begin{theorem}
Define $T_1$ as the first instant where $v[T_1] \in D$, and define $T_2$ as the first instant where $v[T_2] \in W'$. When $g<0$, $v[k_0] \in S$ is necessary but not sufficient for system $\Sigma_3$ to exhibit growing oscillations starting at time $k_0$. These oscillations terminate either at time $T_1$ or $T_2$. \label{main_thm}
 \end{theorem}

\begin{proof} %
 Lemma \ref{basis_lemma} (basis step) establishes that \eqref{base_cond} are necessary for on period of oscillations to occur. By Lemma \ref{inductioN_1neg} (induction step), the system continues to oscillate after the first oscillation period. Lemma \ref{expansion} establishes that the oscillations grow, so the system will eventually land in $D$ or outside $W$. 
 \end{proof} %
 
 Theorem \ref{main_thm} implies that when $g<0$, the only way for $\Sigma_3$ to exhibit oscillations is when $v[k_0] \in S$. To use this theorem, engineers would choose control parameters such that $S=\emptyset$. Further, if each device has a copy of all control parameters, they can reject incoming parameter updates when $S\neq \emptyset$. 
\section{Results and Scalability}
\label{results_section}

\subsection{Parameter Plots}
In this section we use MATLAB's MPT toolbox to plot set $S$ in the $(v_1,v_2)$ space and examine its implications on the ratio of the device control delays. Recall that $S$ is the only voltage region where oscillations can begin, and includes the projection of basis conditions \eqref{base_cond} onto the $(v_1,v_2)$ space.

\begin{figure}[!h]  
    \begin{subfigure}{0.24\textwidth}
         \centering
       \includegraphics[width=\textwidth]{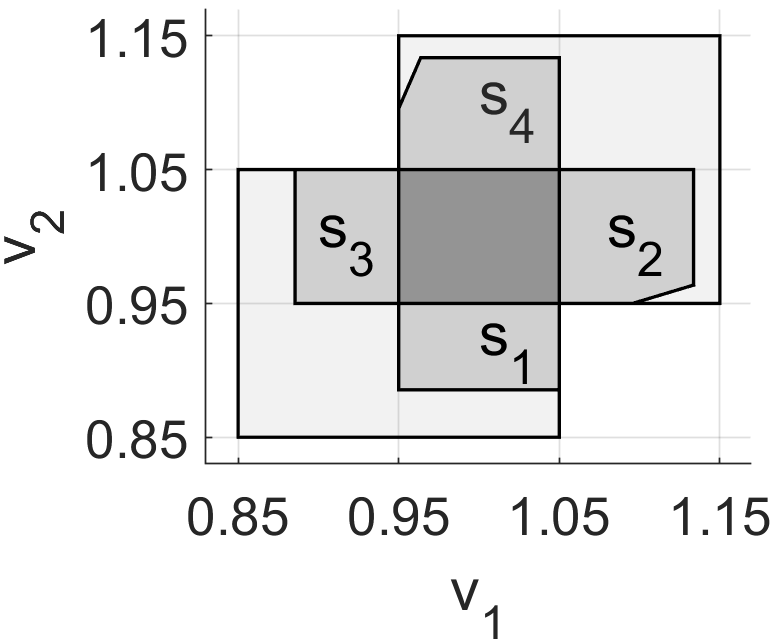}
        \caption{$s_i$ regions when $g=0.5$}
         \label{plotIC_gpos}
     \end{subfigure}
     \hfill
    \begin{subfigure}{0.24\textwidth}
         \centering
       \includegraphics[width=\textwidth]{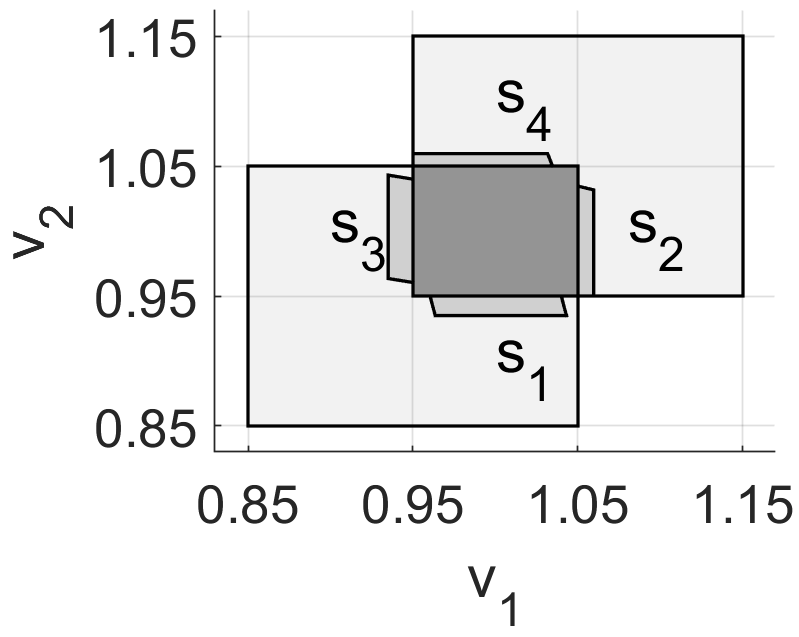}
        \caption{$s_i$ regions when $g=-0.3$}
         \label{plotIC_gneg}
     \end{subfigure}
     \caption{$S$ in medium gray on $(v_1,v_2)$ space; $W$ in light gray}
         \label{plotIC}
\end{figure}
 In Fig. \ref{plotIC} we plot $S$ on the $(v_1,v_2)$ space, and validate that $S$ in Fig. \ref{plotIC_gneg} is indeed a subset of $W_o$ from \eqref{Wo_4device}, as both sets have width of 0.0095 pu. Observe that some corners of $W_o$ are not included in $S$ because trajectories that start there leave $W$ before completing an oscillation period. We also observe that in Fig. \ref{plotIC_gneg} where oscillations can grow, the area of $S$ is smaller compared to when oscillations are damped in Fig. \ref{plotIC_gpos}. It is relieving that the regions where dangerous oscillations could occur are narrow.

    \begin{figure}[!h]  
       \centering 
       \includegraphics[width=.45\textwidth]{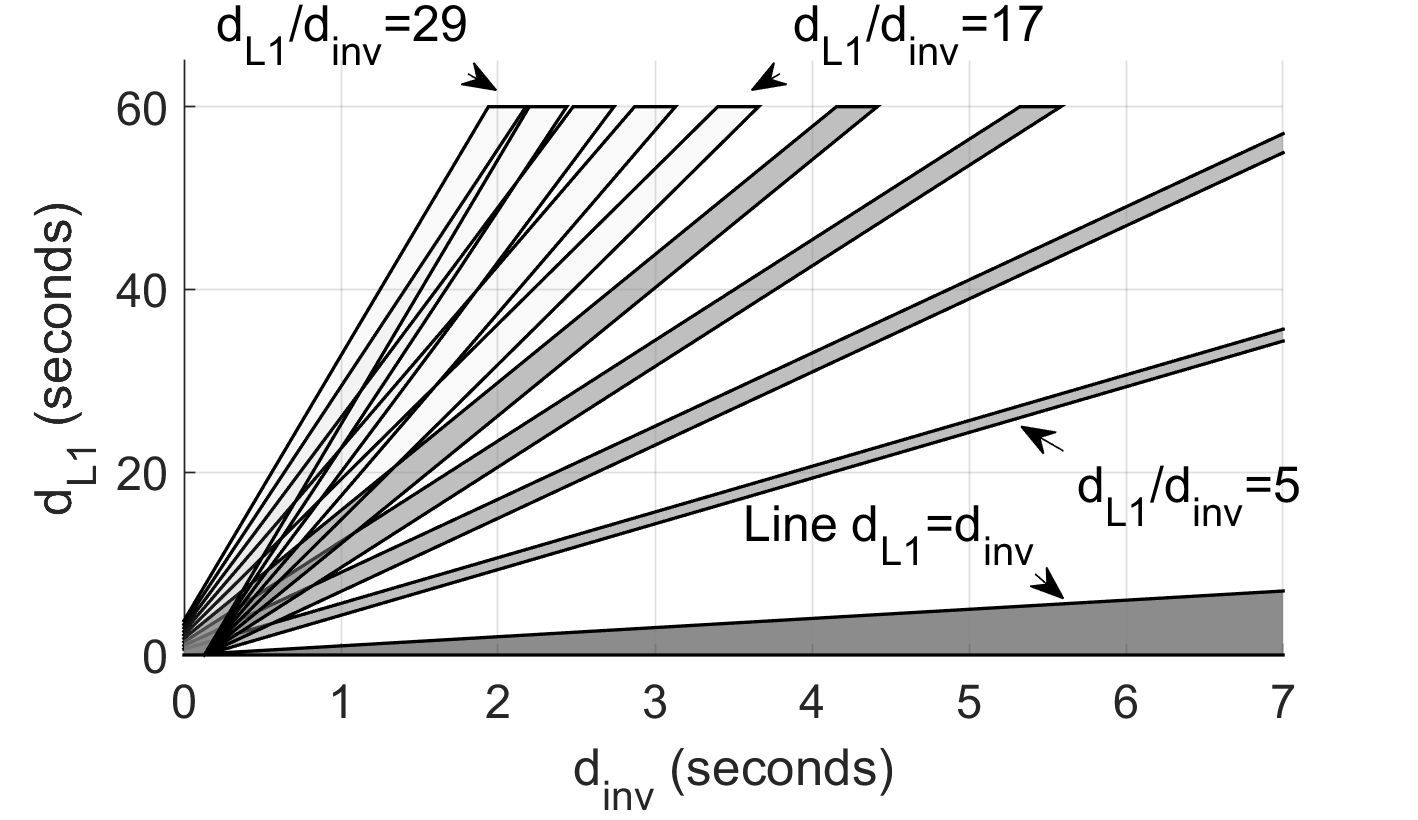}
        \caption{Lines for whether $S=\emptyset$ or $S\neq \emptyset$ on $(d_{L1},d_{inv})$ space; hunting impossible (possible) in light (dark) gray}
        \label{plotDelay} 
    \end{figure}
Next we examine the appropriate timescale separation between the inverters and LTC1 for system $\Sigma_3$ when $g<0$. Because $N_1= \text{floor}(d_{L1}/d_{inv})$ and $N_1$ enters $\Delta v^{inv}$ as the exponent (see \eqref{parametrized}), Each fixed $N_1$ represents the slope of a line through the origin on the $(d_{L1},d_{inv})$ space. As such, varying $d_{L1}$ and $d_{inv}$ such that $N_1$ is fixed will not change the $S$ region on $(v1,v2)$. Therefore, rather than iterate over the $(d_{L1},d_{inv})$ in a meshgrid, we compare the area of set $S$ for lines corresponding to different $N_1$ values. 
 
We set all variables except $(v_1[k_0]$ and $v_2[k_0])$ to the defaults in Table \ref{tab:base_parms}. Instead of $d_{L1}=30$ and $d_{L2}=40$, which have a ratio of $\frac{4}{3}$, we set $N_2=\frac{4}{3}N_1$ since $\frac{N_2}{N_1}=\frac{d_{L2}}{d_{L1}}$. We vary $N_1$ from 5 to 29. For each $N_1$ value, if $S=\emptyset$, then hunting oscillations are impossible for all $v[k_0] \in W$, and we mark the line with slope $N_1$ as light gray in Fig. \ref{plotDelay}. Conversely, if $S \neq \emptyset$ then the $N_1$ line is dark gray. The line $d_{inv}=d_{L1}$ is also marked to only allow the $d_{inv}<d_{L1}$ from Table \ref{tab:base_parms}. We observe that $N_1=17=\text{floor}(d_{L1}/d_{inv})$ is the borderline case where hunting is impossible. 
For example, if the inverters have a 3-second delay, the substation LTC needs at least a 51-second delay to prevent hunting for all IC in $W$.

\subsection{Grids with More than Four Devices}
Radial distribution grids are typically tree graph networks whose root is the substation. Each node has one or more \emph{branch(es)}, which is the tree network connected to that node's edge(s). Suppose there is at most one LTC and at most two actively controlled inverters on each branch of the root node $b$. Then each branch has a system comprised of the substation LTC, the LTC on the branch, and two inverters on the branch. 
One can characterize each system as $\Sigma_3$ and setup conditions to prevent oscillations separately using the methodology of section \ref{probform2_section}. The conditions across all branches must be jointly satisfied, but because the substation LTC is common to all subsystems, as $b$ grows the number of symbolic variables grows by $3b+1$ instead of $4b$.

\section{Conclusion}
\label{conclusion_section}

We have presented a novel hybrid system model for LTCs with inverters on radial distribution circuits. Leveraging the system dynamics, we have derived conditions on the control parameters to guarantee against voltage oscillations created by device hunting. The conditions inform the design of appropriate parameters, such as the minimal timescale separation of control delays between LTCs and inverters. The conditions also pave the way for implementing on-board certificates that guard against malicious firmware updates of control parameters. 

Future work will more formally investigate the types of events that widen the voltage difference enough for hunting to be possible. Additionally, the relationships between parameters captured by the conditions derived here will be examined in more detail. Finally, the behavior of the system in abnormal but not impossible operating states will be explored.

    \bibliography{citations.bib}
    \bibliographystyle{ieeetr}
\fi
\ifdefined\arxiv %
We include proofs of the main theorem and lemmas in-line, and refer readers to the appendix for the minor proofs.

\section{Appendix}

\noindent \textbf{Lemma \ref{lem_obv_2LTC}} %
    If $\bar{v}_L>2\varepsilon$, system $\Sigma_1$ will have marginally stable oscillations for all time when any $v_1[0] \in M(D,c)$ or $v_2[0] \in M(D,c)$ where $c=\bar{v}_L-2\varepsilon>0$ %
\begin{proof} 
    If $v_1[0] \in M(D,c)$, LTC1 taps, causing $\bar{v}_L$ to overshoot the deadband of width $2\varepsilon$ and land outside $D$. After a delay $d_{L1}$, LTC1 will tap in the opposite direction, landing at the system IC. Then these actions repeat, causing $v_1$ and $v_2$ to oscillate with constant amplitude for all time. When $v_2[0] \in M(D,c)$ we get the same oscillatory behavior. 
\end{proof} 

\noindent \textbf{Lemma \ref{lem_2LTC}.} %
    If $\bar{v}_L\leq 2 \varepsilon$, $v[T] \in W_o$, and $v^{\text{diff}}<2\varepsilon-\bar{v}_L$, system $\Sigma_1$ will have marginally stable oscillations starting at time $T$
\begin{proof}
    First consider an IC in $W_o$ (i.e. $T=0$) where $v_2[0]$ is an undervoltage. Overshooting the deadband with an LTC tap can be represented with
    \begin{align}
        v_2[0]&<v^- ~\text{start with undervoltage} \nonumber \\
        v_1[0]+\bar{v}_L&>v^+. \nonumber %
    \end{align}
     Combining the above two equations gives $v^{\text{diff}}>2\varepsilon-\bar{v}_L$. These conditions for $v_2[0]$ being an \emph{overvoltage} would yield the same $v^{\text{diff}}>2\varepsilon-\bar{v}_L$. Now consider where our IC is not in $W_o$ but at time $T$ we land in $W_o$. Because the system has no internal memory states, oscillations will begin after this nonzero $T$ as if the zero-start time was at $T$.
\end{proof}

\noindent \textbf{Lemma \ref{QV_int_diverge}.} %
    If $G \prec0$, system $\Sigma_2$ given by \eqref{QV_int_closedloop} has $v \rightarrow \pm \infty$. %
\begin{proof} 
    Let $F=-G$. Because $G \prec0$ and is diagonal, $F$  is diagonal and positive definite. The proof of Theorem 3.1 in \cite{Helou} shows that
    $eig(XS)=eig(S^{\frac{1}{2},\top}XS^{\frac{1}{2}}) \in (0,2)$ for symmetric $X$ and diagonal positive definite $S$. Thus $eig(XF)=eig(F^{\frac{1}{2},\top}XF^{\frac{1}{2}}) \in (0,2)$. Then all $eig(I-XG)=eig(I+XF)=1+eig(XF)=1+(0,2)>1$. Thus the spectral radius of $(I-XG)$ is greater than 1, so $e \rightarrow \infty$, and $v \rightarrow \infty$ or $v \rightarrow -\infty$. 
\end{proof}

\noindent \textbf{Lemma \ref{homogeneous}.} %
    The coupling effect of a single inverter actuating at node $i$ on the voltage at node $j$ is damped by a factor of $\eta$. That is, $\Delta v_j^{inv}(X_{ij},k+N,k,v_i[k])=\eta \Delta v_i^{inv}(X_{ii},k+N,k,v_i[k])$.
\begin{proof}
    From the update equations \eqref{v_update_scalar1} and \eqref{v_update_scalar2},
     \begin{align*}
         \Delta v_i^{inv}(X_{ii},k+1,k,v_i[k]) &=-\chi(v_i[k]-v^{ref})\\
        \Delta v_j^{inv}(X_{ij},k+1,k,v_i[k]) &=-\eta\chi(v_i[k]-v^{ref})\\
         &=\eta\Delta v_i^{inv}(X_{ii},k+1,k,v_i[k])
    \end{align*}
    By Remark \ref{remark_deltaV_sign}, the voltage change at each timestep over  $[0 ~N]$ is additive in the same direction. If we consider a duration over two timesteps,
     \begin{multline}
        \Delta v_j^{inv}(X_{ij},k+2,k,v_i[k]) =v[k+2]-v[k+1]+v[k+1]-v[k]\\=\Delta v_j^{inv}(\chi,k+2,k+1,v_i[k+1])+\Delta v_j^{inv}(\chi,k+1,k,v_i[k])\\=-\eta\chi(v_i[k+1]-v^{ref})-\eta\chi(v_i[k]-v^{ref})\\=\eta\Delta v_i^{inv}(X_{ii},k+2,k,v_i[k]) \nonumber
    \end{multline}
\end{proof}

\noindent \textbf{Lemma \ref{expansion}.} %
    If system $\Sigma_3$ with $g<0$ has oscillations, $v^{\text{diff}}$ increases after each oscillation period %
\begin{proof}
    We consider the case that the IC is below the deadband. We express the voltages after the first sequence of modes and compare $v^{\text{diff}}[k_0]$ to $v^{\text{diff}}[k_4]$
    \begin{align}
        v_1[k_4]=v_1[k_0]+\eta \Delta v_2^{inv}(k_0+N_2,k_0,v_2[k_0])+\bar{v}_L \nonumber \\+\Delta v_1^{inv}(k_0+N_1,k_0,c)-\bar{v}_L \notag\\
        v_2[k_4]=v_2[k_0]+ \Delta v_2^{inv}(k_0+N_2,k_0,v_2[k_0])+\bar{v}_L\nonumber \\+\eta \Delta v_1^{inv}(k_0+N_1,k_0,c)-\bar{v}_L \nonumber
    \end{align}
    where $c \coloneqq v_1[k_0]+\eta \Delta v_1^{inv}(k_0+N_2,k_0,v_2[k_0])-\bar{v}_L$. Subtracting the above two equations, we have
    \begin{align}
        \begin{split}
        v^{\text{diff}}&[k_4]= v^{\text{diff}}[k_0]+\eta \Delta v_2^{inv}(k_0+N_2,k_0,v_2[k_0]) \notag\\ &+\Delta v_1^{inv}(k_0+N_1,k_0,c) -\Delta v_2^{inv}(k_0+N_2,k_0,v_2[k_0])\notag\\&-\eta \Delta v_1^{inv}(k_0+N_1,k_0,c)-\bar{v}_L
        \end{split}\notag
    \end{align}
    which simplifies to
    \begin{align}
        \begin{split}
                v^{\text{diff}}[k_4] &- v^{\text{diff}}[k_0]=(\eta-1)(\Delta v_2^{inv}(v_2[0],N_2)\\&-\Delta v_1^{inv}(k_0+N_1,k_0,c)) \label{contract1}
        \end{split}
    \end{align}
    
    The damping factor has $0<\eta<1$, so $(\eta-1)$ is negative. $v_2[k_0]$ is below the deadband, so $\Delta v_2^{inv}(N_2,v_2[k_0])<0$, and $v_1[k_2]$ is above the deadband, so $\Delta v_1^{inv}(c,N_1)$ is positive. Together, \eqref{contract1} is always positive. Thus $v^{\text{diff}}[k_4]>v^{\text{diff}}[k_0]$.
    
    If we repeat this process for the case of the IC being an overvoltage, we get the same final equation \eqref{contract1}.
\end{proof}

\section{Acknowledgement}
The authors would like to express their sincere gratitude to Professor Murat Arcak and Lawrence Berkeley National Lab Scientist Daniel Arnold for their insightful comments.
    \bibliography{citations.bib}
    \bibliographystyle{ieeetr}
\fi

\end{document}